\documentclass[11pt, oneside]{article}   	
\usepackage{geometry}                		
\usepackage{graphicx}				
\usepackage{amssymb}
\usepackage{amsmath}
\usepackage{amsthm}

\usepackage{enumerate} 

\newcommand{\sgn}{\mathop{\mathrm{sgn}}} 
\newcommand{\ac}{\mathrm{ac}} 
\newcommand{\rad}{\mathrm{rad}}
\newcommand{\re}{\mathop{\mathrm{Re}}} 
\newcommand{\im}{\mathop{\mathrm{Im}}} 
\newcommand{\tr}{\mathop{\mathrm{Tr}}}
\newcommand{\supp}{\mathop{\mathrm{supp}}}

\newcommand{\Z}{\mathbb{Z}}
\newcommand{\R}{\mathbb{R}} 
\newcommand{\C}{\mathbb{C}}

\numberwithin{equation}{section}

\theoremstyle{plain}
\newtheorem{thm}{Theorem}[section]
\newtheorem{proposition}[thm]{Proposition}
\newtheorem{lemma}[thm]{Lemma} 
\newtheorem{corollary}[thm]{Corollary}

\theoremstyle{definition} 

\newtheorem{defn}[thm]{Definition}

\newtheorem{assump}[thm]{Assumption}

 \newtheorem{remark}[thm]{Remark}

 \newtheorem*{remarks*}{Remarks}
\newtheorem*{remark*}{Remark}

\title{Orthonormal Strichartz estimates for Schrödinger operator and their applications to infinitely many particle systems }
\author{Akitoshi Hoshiya\thanks{Graduate School of Mathematical Sciences, The University of Tokyo, 3-8-1 Komaba, Meguro-ku, Tokyo 153-8914, Japan \\
 Email address: hoshiya@ms.u-tokyo.ac.jp}}

\begin{document}
\maketitle

\begin{abstract}
We develop an abstract perturbation theory for the orthonormal Strichartz estimates, which were first studied by Frank-Lewin-Lieb-Seiringer. The method used in the proof is based on the duality principle and Kato's smooth perturbation theory. We also deduce the refined Strichartz estimates for the Schrödinger operator in terms of the Besov space. Finally we prove the global existence of a solution for the Hartree equation with electromagnetic potentials describing the dynamics of infinitely many fermions. This would be the first result on the orthonormal Strichartz estimates for the Schrödinger operator with general time-independent potentials including very short range and inverse square type potentials.
\end{abstract}

\section{Introduction}\label{23110221}

This paper is devoted to the proof of the orthonormal Strichartz estimates for the Schrödinger operator and their applications.  The orthonormal Strichartz estimates are studied by Frank-Lewin-Lieb-Seiringer \cite {FLLS} for the first time. They are the following inequalities: 
\begin{align}
\left \| \sum_{n=0}^ \infty{\nu_n|e^{it\Delta}f_n|^2} \right\|_{L^p_t L^q_x} \lesssim \| \nu_n\|_{\ell^\alpha}.
\label{231115}
\end{align}
Here  $\{f_n\}$ is an orthonormal system in $L^2(\mathbb{R}^d)$ and $p, q, \alpha \in [1, \infty]$ satisfy some conditions specified later. 
 $e^{it\Delta}$ denotes the free Schrödinger propagator.

The aim of this paper is to extend the orthonormal Strichartz estimates to the Schrödinger operator: $H= -\Delta + V$.
We are going to develop an abstract perturbation method based on Kato's smooth perturbation theory \cite{KY} and apply it to potentials of very short range or inverse square type. The magnetic Schrödinger operator satisfying the assumptions in \cite{EGS} is also dealt with.

After proving the orthonormal Strichartz estimates for the Schrödinger operator, we verify 
the refined Strichartz estimates in terms of the Besov spaces:
\[ \|e^{-itH}P_{\ac}(H)u\|_{L^{2p}_t L^{2q}_x} \lesssim \|u\|_{\dot{B}^0_{2,\frac{4q}{q+1}}} .\]Here $P_{\ac}(H)$ denotes the orthonormal projection onto the absolutely continuous subspace. For the free case: $H=-\Delta$, this is proved in \cite{FS}. And finally we prove the global existence of a solution for the Hartree equation describing infinitely many particles:
\[
\left\{
\begin{array}{l}
i\partial_t \gamma=[H+w*\rho_{\gamma},\gamma], \\
\gamma(0)=\gamma_0,
\end{array}
\right.
\]     
where $H= (D+A(x))^2 +V(x)$. Here $\gamma$ is an operator-valued function and $\rho_{\gamma}$ denotes its density function. Global existence for $H=-\Delta$ is also proved in \cite{FS}.

The orthonormal Strichartz estimates \eqref{231115} is important from two perspectives.
First it is a refinement of the usual Strichartz estimates: 
\[ \|e^{-it\Delta}u\|_{L^{2p}_t L^{2q}_x} \lesssim \|u\|_2 , \]
 which has been used in the analysis of the nonlinear Schrödinger equations for about fifty years. 
This estimate was first deduced by Strichartz \cite{St} for $p=q$ and has been extended to more general $p,q$ by Ginibre-Velo \cite{GV} and Yajima \cite{Y} except endpoint case. Finally, Keel-Tao \cite{KT} proved the endpoint estimates and gave an abstract method to prove the Strichartz estimates from dispersive estimates. Note that if we take $f_n =0$ for $n\ge 1$ in the orthonormal Strichartz estimates, we obtain the usual one.
Second it is used in the analysis of the Hartree equations describing infinitely many particle system of fermions. This was first studied in \cite{LS} and the local or global well-posedness was proved under various conditions. For higher order or fractional cases,  \cite{BLN} proved the local existence of a solution. Recently the asymptotic stability of stationary solutions is considered in \cite{Ha}.

After \cite{FLLS}, the orthonormal Strichartz estimates have been extended by Frank-Sabin \cite{FS} for more general $p, q, \alpha$. Bez-Hong-Lee-Nakamura-Sawano \cite{BHLNS} extended the estimates for the case when $\{f_n\}$ is an orthonormal system in $\dot{H}^s$.  In Bez-Lee-Nakamura \cite{BLN}, the orthonormal Strichartz estimates for the fractional Schrödinger, wave and Klein-Gordon equations are also treated.

As far as the author knows, the orthonormal Strichartz estimates for the Schrödinger operator $H$ have been considered only for the special operator: the special Hermite operator, the $(k, a)$-generalized Laguerre operator (including harmonic oscillators) and the Dunkl operator. In \cite{Ms}, the special Hermite operator is discussed. In \cite{MS}, the local in time orthonormal Strichartz estimates for the $(k, a)$-generalized Laguerre operator and the global in time orthonormal Strichartz estimates for the Dunkl operators are proved under some restriction (harmonic oscillator is also treated in \cite{BHLNS}). On these operators, the recent result \cite{TT} proves the usual Strichartz estimates under more general assumptions. As far as we are aware of, we do not know any other results except for the above on the global in time orthonormal Strichartz estimates for the general Schrödinger operator. On the other hand, the orthonormal Strichartz estimates for $H= -\Delta +V(t, x)$ where $V \in L^{\mu}_tL^{\nu}_x$ for $\mu<\infty, \frac{2}{\mu} +\frac{d}{\nu} =2$ are recently considered in \cite{Ha}.

For the usual Strichartz estimates, such generalization has been done by many researchers. In \cite{RS}, Rodnianski-Schlag proved the Strichartz estimates for $H=-\Delta +V$ when $|V(x)| \lesssim \langle x \rangle ^{-2-\epsilon}$ except endpoint. The above potentials are said to be very short range. They also gave an abstract perturbation method for the Strichartz estimates, which enables us to deduce the Strichartz estimates from the Kato smoothing estimates. In \cite{BM}, Bouclet-Mizutani proved the Strichartz estimates for inverse square type potentials including endpoint and double endpoint estimates. They also generalized the perturbation method above including double endpoint cases. In [M1], the Strichartz estimates for $H$ satisfying $|V(x)| \lesssim \langle x \rangle ^{-\mu}$ for some $\mu \in (0, 2)$ and further assumptions are considered. In that case, it is difficult to use the abstract perturbation method. Hence Mizutani used the method of microlocal analysis, which has been used to prove the Strichartz estimates for variable coefficient operators. 

For the magnetic Schrödinger operator: $H= (D+A(x))^2 +V(x)$, $D=-i\nabla$,  the Strichartz estimates are also proved by \cite{EGS} except endpoint using the above perturbation method. The endpoint case is discussed in \cite{DFVV}.

For variable coefficient operators: $H=-\partial_i a_{ij}\partial_j +V$, where $a_{ij}$ is an asymptotically flat metric, the Strichartz estimates are considered in \cite{MMT} when $|V(x)|\lesssim \langle x \rangle ^{-2} (\log \langle x \rangle )^{-2}$. In \cite{MY1}, they considered inverse square type potentials using the above perturbation method and \cite{MMT}.

It is also possible to consider the orthonormal Strichartz estimates for the fractional, higher order, wave and Klein-Gordon equations with potentials. They would be discussed in the forthcoming paper including magnetic potentials.

The paper is organized as follows. In Section~\ref{2311152141} we explain our main results and give some remarks. In Section \ref{2311022103}, we collect some lemmas used in the proof of our abstract perturbation method (Theorem \ref{2311022111}), e.g., the duality principle, the Christ-Kiselev type lemma and the orthonormal Strichartz estimates for the free Hamiltonian.  Finally in Section 4, we give the proof of main results.

\section{Main Results}\label{2311152141}

Before stating our main results, we give some notations.

\subsubsection*{\textbf{Notations}}

\begin{itemize}
\item
For a measure space $(X, d\mu)$, $ L^p(X)$ denotes the usual $L^p$-space and its norm is $\| \cdot \|_p$. 

\item
 $\mathcal{F}$ denotes the Fourier transform on $\mathcal{S}'$. Here $\mathcal{S}'$ denotes the set of all the tempered distributions.

\item
 For a Banach space $X$, $L^p_t X$ denotes the set of all the measurable $f : \mathbb{R} \rightarrow X$ such that $\|f\|_{L^p_t X} :=(\int_{\mathbb{R}} \|f(t)\|^p_X dt)^{1/p} < \infty$

\item For a self-adjoint operator $H$ and a Borel measurable function $f$, $f(H)$ is defined as $f(H)= \int_{\mathbb{R}} f(\lambda) dE(\lambda)$. Here $E(\lambda)$ is the spectral measure associated to $H$.

\item For a self-adjoint operator $H$, $P_{\ac}(H)$ denotes the projection onto the absolutely continuous subspace $\mathcal{H}_{\ac}$. Here $u \in \mathcal{H}_{\ac}$ iff $\langle E(\lambda)u, u \rangle$ is absolutely continuous with respect to $\lambda$.

\item For a normed space $X$, $\mathcal{B}(X)$ denotes the set of all the bounded operators on $X$.
\end{itemize}

\subsubsection*{\textbf{Kato's smooth perturbation theory}}

We also need to recall Kato's smooth perturbation theory. 
The next theorem is one of the key lemmas in the smooth perturbation theory by Kato \cite{KY}. 

\begin{thm}[\cite{KY}]\label{2311152150}

Suppose $\mathcal{H}$ is an arbitrary Hilbert space. Let $A$ be a densely defined closed operator and $H$ be a self-adjoint operator on $\mathcal{H}$. Then the following are equivalent:
\begin{enumerate}[(1)]
\item
$\|Ae^{-itH}u\|_{L^2_t \mathcal{H}} \lesssim \|u\|_\mathcal{H}$ holds for any $u \in \mathcal{H}$.
\item
$ \sup_{z\in \mathbb{C} \backslash \mathbb{R}}|\langle \im(H-z)^{-1} A^* u, A^* u \rangle| \lesssim \|u\|^2_{\mathcal{H}} $ holds.
\end{enumerate}
In particular, if $ \sup_{z\in \mathbb{C} \backslash \mathbb{R}}\| A(H-z)^{-1} A^*\|_{\mathcal{B}(\mathcal{H})} \lesssim 1$, we have (1) and (2).
\end{thm}
The above smoothing estimate is called the Kato smoothing estimate or the Kato-Yajima estimate.

\begin{defn}
If $A$ satisfies the condition in Theorem~\ref{2311152150}, we say that $A$ is $H$-smooth.
\end{defn}

\noindent
Now we give our main results. First we prove the following abstract perturbation theorem.

\begin{thm}\label{2311022111}

Let $(X ,d\mu)$ be a $\sigma$-finite measure space and $\mathcal{H}:=L^2(X)$. Assume that for self-adjoint operators $H , H_0$ and densely defined closed operators $Y, Z $ on $\mathcal{H}$, 
$H=H_0 + V ,V=Y^*Z$ holds in the form sense, i.e.
\[D(H) \cup D(H_0) \subset D(Y) \cap D(Z) \quad \text{and}\quad  \langle Hu,v \rangle=\langle H_0 u,v\rangle+\langle Zu,Yv\rangle \] 
for all $u, v \in D(H) \cap D(H_0)$.
We also assume $Y$ is $H_0$-smooth and $ZP_{ac}(H)$ is $H$-smooth.
If we have 
\[ \left\| \sum_{n=0}^ \infty{\nu_n|e^{-itH_0}f_n|^2} \right\|_{L^p_t L^q_x} \lesssim \| \nu_n\|_{\ell^\alpha}\] for some $p, q \in [1,\infty] $, $ \alpha \in(1,\infty)$ and all orthonormal systems $\{f_n\}$ in $\mathcal{H}$,
then we also have 
\[ \left\| \sum_{n=0}^ \infty{\nu_n|e^{-itH}P_{ac}(H)f_n|^2} \right\|_{L^p_t L^q_x} \lesssim \| \nu_n\|_{\ell^\alpha}.\]
\end{thm}

As a corollary of this theorem, we have the following concerning the Schr\"odinger operator on $\mathbb{R}^d$.

\begin{corollary}[Very short range potentials]\label{11161517}

Let $H=H_0 +V$, $H_0 = -\Delta$, $V:\mathbb{R}^d \rightarrow \mathbb{R}$,
$|V(x)| \lesssim \langle x \rangle^{-2-\epsilon}$.
We assume $zero$ is neither an eigenvalue nor a resonance of $H$.
Then we have 
 \[ \left\| \sum_{n=0}^ \infty{\nu_n|e^{-itH}P_{ac}(H)f_n|^2} \right\|_{L^p_t L^q_x} \lesssim \| \nu_n\|_{\ell^\alpha}\]  for $p ,q,\alpha$  satisfying either of the following.
\begin{enumerate}[(S1)]
\item \label{2311152155}
$d  \ge1 ,p, q \in [1,\infty] ,2/p +d/q =d ,q\in[1,{\frac{{d+1}}{d-1})} ,\alpha=\frac{2q}{{q+1}}$;

\item \label{2311152156}
$d\ge3 ,2/p +d/q =d ,q\in(\frac{d+1}{d-1},\frac{d}{d-2}) ,\alpha < p$.
\end{enumerate}
\end{corollary}

Note that $zero$ is called a resonance iff there exists $u \ne 0$ satisfying $-\Delta u+Vu=0 , u\in L^2_{-\alpha}$ for some $\alpha>1/2 $. Here $L^2_{\alpha}$ is the weighted $L^2$ space and its norm is defined by $\| u \|_{L^2_{\alpha}}=\|\langle x \rangle ^\alpha u\|_{2}$.

\begin{corollary}[Potentials of inverse square type]\label{ooo}  
Let $H= H_0 +V$, $H_0 = -\Delta$, $V:\mathbb{R}^d \rightarrow \mathbb{R}$, $d \ge3$.
We assume $V\in X^\sigma_d :=\{ V:\mathbb{R}^d \rightarrow \mathbb{R} \mid |x|V \in M^{d,2\sigma} ,x\cdot \nabla V \in M^{d/2,\sigma}  \} $ for some $ \sigma \in (\frac{d-1}{2},\frac{d}{2})$,
and 
\begin{align*}
\langle Hf,f\rangle  \gtrsim \| \nabla f \|^2_2,
\quad 
\langle (-\Delta-V-x\cdot \nabla V)f,f \rangle \gtrsim \| \nabla f \|^2_2 
\end{align*}
hold for any $f\in C^\infty_0(\mathbb R^d)$.
Then we have 
\[\left\| \sum_{n=0}^ \infty{\nu_n|e^{-itH}f_n|^2} \right\|_{L^p_t L^q_x} \lesssim \| \nu_n\|_{\ell^\alpha}\]
for $p, q, \alpha$ satisfying $(S1)$ or $(S2)$. 
In the definition of $X^{\sigma} _d$, $M^{p, q} $ denotes the Morrey-Campanato space. See Section2 of \cite{BM} for its definition and properties.
\end{corollary}

\begin{remark}

Under the conditions in Corollary \ref{ooo}, $P_{ac}(H)=I$ holds. Thus we omit $P_{ac}(H)$ in the estimate.
Typical examples of potential $V$ in Corollary \ref{ooo} are inverse square potentials: $V= c|x|^{-2} ,c>-(d-2)^2/4$.
Note that $(d-2)^2/4$ is the best constant of the Hardy inequality. Also see \cite{BM} for more details about properties of $H$ in Corollary \ref{ooo}.
\end{remark}

\begin{remark}

In Corollary \ref{ooo}, the critical inverse square potential $V=-(d-2)^2|x|^{-2}/4$ is not included. For  the usual Strichartz estimates for $H=-\Delta -(d-2)^2|x|^{-2}/4$, the non-endpoint Strichartz estimates are proved in \cite{Su}. However, in \cite{M2}, Mizutani proved that the endpoint Strichartz estimates does not hold. We prove the following orthonormal Strichartz estimates for $H$ when $(p, q)$ is close to $(\infty, 1)$ .

\end{remark}

\begin{thm}[Critical inverse square potential]\label{1113026}

Let $H=-\Delta -(d-2)^2|x|^{-2}/4$ and $d \ge 3$. $P_{\rad}$ denotes the orthogonal projection onto the subspace of radial functions: $P_{\rad}: L^2(\R^d) \rightarrow L^2_{\rad}(\R^d), f \mapsto \frac{1}{|S^{d-1}|} \int_{S^{d-1}} f(|x|\theta) d\sigma(\theta)$. Then we have the following estimates:

\begin{enumerate}
\item[(1)] For $p, q, \alpha$ satisfying $(S1)$ or $(S2)$,  
\[ \left\| \sum_{n=0}^ \infty{\nu_n|e^{-itH}P_{\rad}^{\perp} f_n|^2} \right\|_{L^p_t L^q_x} \lesssim \| \nu_n\|_{\ell^\alpha}\]
holds. Here $P_{\rad}^{\perp} = I - P_{\rad}$.
\item[(2)] For $p, q, \alpha$ satisfying $d=3, 1\le q < \frac{3d}{3d-4}$ and $(S1)$, 
\[ \left\| \sum_{n=0}^ \infty{\nu_n|e^{-itH}P_{\rad} f_n|^2} \right\|_{L^p_t L^q_x} \lesssim \| \nu_n\|_{\ell^\alpha}\]
holds.
\end{enumerate}
Hence
\[ \left\| \sum_{n=0}^ \infty{\nu_n|e^{-itH} f_n|^2} \right\|_{L^p_t L^q_x} \lesssim \| \nu_n\|_{\ell^\alpha}\]
holds for $p, q, \alpha$ satisfying the above condition in $(2)$.
\end{thm}

\begin{remark}

We have seen the orthonormal Strichartz estimates for very short range potentials or inverse square type potentials.  It is also interesting to consider the case when $V$ is a slowly decaying potential (i.e.\  $|V(x)|\lesssim \langle x \rangle^{-\mu} $ for some $\mu \in(0,2)$). In that case the usual Strichartz estimates are known to hold under the condition that $V$ is $C^{\infty}$, $|\partial^{\alpha} V(x)| \lesssim \langle x \rangle^{-\mu-|\alpha|}$, $V(x) \gtrsim \langle x \rangle^{-\mu}$ and $-x \cdot \nabla V(x) \gtrsim \langle x \rangle^{-\mu}$ for large $x$ (\cite{M1}). Although there is no positive result concerning the orthonormal Strichartz estimates for slowly decaying potentials, we have the following counterexample, which was found by Goldberg-Vega-Visciglia \cite{GVV} for the usual Strichartz estimates:
Assume $V\in C^3(\mathbb{R}^d  \backslash \{0\} ;\mathbb{R})$, $V(x)=|x|^{-\mu}V(\frac{x}{|x|})$ for some $\mu\in(0,2)$. If $\min_{|x|=1}V(x)=0$ and the  minimum point is non-degenerate, then the orthonormal Strichartz estimates for $H = -\Delta +V$ fail for all $p ,q,\alpha$ satisfying (S1) or (S2) except the trivial case: $(p, q)=(\infty, 1)$.
Note that $\sigma(H) = \sigma_{\ac}(H) = [0, \infty)$ holds in the above example since $V$ is a repulsive potential.
\end{remark}

Next we consider the magnetic Schrödinger operator.
\begin{thm}\label{11102255}
Assume $d \ge3$, $A:\R^d \rightarrow \R^d, V:\R^d \rightarrow \R$ satisfy $|A(x)| + |\langle x \rangle V(x)| \lesssim \langle x \rangle^{-1-\epsilon}$, $\langle x \rangle^{1+\epsilon'}A(x) \in \dot{W}^{1/2, 2d}$, $A \in C^0(\R^d)$ for some $0<\epsilon'<\epsilon$. Let $H = -\Delta +A(x)\cdot D + D\cdot A(x) + V(x)$ or $H = (D +A(x))^2 +V(x)$. Here $D=-i\nabla$. If zero is neither an eigenvalue nor a resonance of $H$, 
\[ \left\| \sum_{n=0}^ \infty{\nu_n|e^{-itH}P_{\ac}(H) f_n|^2} \right\|_{L^p_t L^q_x} \lesssim \| \nu_n\|_{\ell^\alpha}\]
holds for $p, q, \alpha$ satisfying $(S1)$ or $(S2)$.

\end{thm}

\begin{remark}
Concerning the usual Strichartz estimates for the magnetic Schrödinger operator, \cite{EGS} proved the Strichartz estimates except endpoint under the same assumptions as above. The endpoint Strichartz estimates are considered in \cite{DFVV} but their assumptions are different from ours.

\end{remark}

When potential $V$ is time-dependent, the orthonormal Strichartz estimates for $H=-\Delta+V(t, x)$ are considered in \cite{Ha}. By using the results in this paper, it is also possible to consider the Hamiltonian like $\tilde{H}(t) = H+V(t, x)$. Here $H=(D+A(x))^2 +V(x)$ denotes the time-independent Hamiltonian. Actually, we can prove the following theorem. The proof is done by just substituting $H$ for $-\Delta$ in the argument of \cite{Ha}, so we omit the details.

\begin{thm}
Let $\tilde{H}(t) = H+V(t, x)$ and $H=(D+A(x))^2 +V(x)$. We assume $H$ is as in 
Corollaries~\ref{11161517}, \ref{ooo}, Theorems~\ref{1113026} or \ref{11102255} and $V(t, x) \in L^{\mu}_t L^{\nu}_x$ for $\mu \in [1, \infty ),\  \nu \in[1, \infty],\  2/\mu + d/\nu =2$. Furthermore we assume the spectrum of $H$ is absolutely continuous. Then the orthonormal Strichartz estimates hold for the propagator $U(t)$ of\ \ $i\partial_t u = \tilde{H}u$ as follows:
\[ \left\| \sum_{n=0}^ \infty{\nu_n| U(t)f_n|^2} \right\|_{L^p_t L^q_x} \lesssim \| \nu_n\|_{\ell^\alpha}.\]
Here $p, q, \alpha$ are from Corollaries~\ref{11161517}, \ref{ooo}, Theorems~\ref{1113026} or \ref{11102255}, respectively. 
\end{thm}

Next we give a refinement of the usual Strichartz estimates in terms of the Besov spaces.

\begin{corollary}\label{jjj}

Let $p ,q,\alpha$ as in (S1) and $d \ge 2$. 
Let $H$ denote the Schr\"odinger operator in Corollaries \ref{11161517} or \ref{ooo}. 
Furthermore, we assume $V\ge0$. Then we have
\[ \|e^{-itH}P_{\ac}(H)u\|_{L^{2p}_t L^{2q}_x} \lesssim \|u\|_{\dot{B}^0_{2,\frac{4q}{q+1}}} \]
Here $\dot{B}^0_{p,q}$ denotes the homogeneous Besov spaces of order 0. Its norm is defined by $\| u\|_{\dot{B}^0_{p,q}}= \| \phi_j(D)u\|_{\ell^q_j L^p_x}$, $\phi_j(D) = \mathcal{F}^* \phi_j(\xi)\mathcal{F}$, 
where $\phi_j(\xi)$ is a homogeneous Littlewood-Paley decomposition.
\end{corollary}

\begin{corollary}\label{11191637}
Let $p ,q,\alpha$ as in (S1), $H$ as in Theorem \ref{11102255}. We also assume 
\[ \|V_{-}\|_{K_d} < c_{d} := \frac{\pi^{d/2}}{\Gamma(d/2 -1)}, \quad  \langle Hu, u \rangle \approx \langle H_{0} u, u \rangle .\]
Then \[ \|e^{-itH}u\|_{L^{2p}_t L^{2q}_x} \lesssim \|u\|_{\dot{B}^0_{2,\frac{4q}{q+1}}} \]
holds. Here $V_{-} = \min \{0, V\}$, $K_d$ is the Kato class (see Definition \ref{11191957}).
\end{corollary}

\begin{remark}
Since $\frac{4q}{q+1}>2$ holds under the conditions in Corollary \ref{jjj} and Corollary \ref{11191637}, we have $\|u\|_{\dot{B}^0_{2,\frac{4q}{q+1}}}\lesssim\|u\|_2$. Hence Corollary \ref{jjj} is a refinement of the usual Strichartz estimates:
\[ \|e^{-itH}P_{\ac}(H)u\|_{L^{2p}_t L^{2q}_x} \lesssim \|u\|_2 .\]
\end{remark}

Before stating the result on the nonlinear problem, we need the definition of the Schatten space.

\begin{defn}[Schatten space]
Let $\mathcal{H}_1$ and $\mathcal{H}_2$ be a Hilbert space. For a compact operator 
$A : \mathcal{H}_1 \rightarrow \mathcal{H}_2$, the singular values $\{ \mu_n \}$ of $A$  is defined as the set of all the eigenvalues of $(A^*A)^{1/2}$. The Schatten space $\mathfrak{S}^{\alpha} (\mathcal{H}_1 \rightarrow \mathcal{H}_2)$ for $\alpha \in [1, \infty]$ is the set of all the compact operators: $\mathcal{H}_1 \rightarrow \mathcal{H}_2$  such that its singular values belong to $\ell^{\alpha}$. Its norm is defined by the $\ell^\alpha$ norm of the singular values. We also use the following notation for simplicity: $\|S\|_{\mathfrak{S}^{\alpha} (\mathcal{H}_0)} := \|S\|_{\mathfrak{S}^{\alpha} (\mathcal{H}_0 \rightarrow \mathcal{H}_0)}$.
\end{defn}

\noindent
For $A \in \mathcal{B}(L^2 (X))$, where $X$ is a $\sigma$-finite measure space, $\rho_A (x) : = k_A (x, x)$ denotes the density function of $A$. Here $k_A (x, y)$ is the integral kernel of $A$. We sometimes write $\rho (A)$.

As an application of the orthonormal Strichartz estimates, we prove the global existence of the solution for the Hartree equations for infinitely many fermions:
\begin{align}
\left\{
\begin{array}{l}
i\partial_t \gamma=[H+w*\rho_{\gamma},\gamma] \\
\gamma(0)=\gamma_0
\end{array}
\right.
\tag{H}\label{2311152204}
\end{align}
Here, $\gamma$ is an operator-valued function and $\rho_{\gamma}$ is its density function. For functions $f$ and $g$, $f*g$ denotes the convolution with respect to the space variable. 
\begin{thm}\label{11102301}

Let $d,p,q$ be as in (S1), $w \in L^{q'} (\R^d)$, $H$ be as in Corollaries \ref{11161517}, \ref{ooo} or Theorem \ref{11102255}. We also assume that the spectrum of $H$ is absolutely continuous. Then for all self-adjoint $\gamma_0 \in {\mathfrak{S}}^{\frac{2q}{q+1}} (L^2 (\R^d))$, we have a unique global solution $\gamma$ to \eqref{2311152204} satisfying $\gamma\in C_t (\mathbb{R} ;{\mathfrak{S}}^{\frac{2q}{q+1}} (L^2 (\R^d)))$, $\rho_{\gamma}\in L^p_{\mathrm{loc},t} L^q_x$.
Note that $f \in L^p_{\mathrm{loc},t} L^q_x$ iff $f \in L^p(K)_t L^q_x$ for all the compact sets $K\subset \mathbb R$.
\end{thm}

\begin{remark}

The Hartree equation \eqref{2311152204} describes infinitely many fermions under the condition that there is a time-independent spatially decaying electromagnetic potential. The corresponding $N$ particles system is described as
\[
\left\{
\begin{array}{l}
i\partial_t u_j=Hu_j+w*(\sum_{k=1}^{k=N} |u_k|^2)u_j \\
u_j(0)= u_{0,j}
\end{array}
\right.
\]     
for $j=1,2,\dots,N$.
See \cite{LS} for more details about its physical background when $H=-\Delta$. 
\end{remark}

\section{Preliminaries}\label{2311022103}
In this section, we collect some notations, definitions and lemmas used in the following section.
\subsubsection*{\textbf{Schatten spaces}}

For the proof of properties associated to the Schatten spaces, see \cite{S}, \cite{Ha}. 
We often use Hölder's inequality for the Schatten space:
\[\|ST\|_{\mathfrak{S}^{\alpha} (\mathcal{H}_0 \rightarrow \mathcal{H}_2)} \lesssim \|S\|_{\mathfrak{S}^{\alpha_1} (\mathcal{H}_1 \rightarrow \mathcal{H}_2)} \cdot \|T\|_{\mathfrak{S}^{\alpha_2} (\mathcal{H}_0 \rightarrow \mathcal{H}_1)}\]
Here $1/\alpha = 1/\alpha_1 + 1/\alpha_2$.
We use the following notations for the Schatten spaces: 
\begin{gather*}
\ \mathfrak{S}^{\alpha}_{t, x}= \mathfrak{S}^{\alpha} (L^2_{t, x}) ,
\ \ 
\mathfrak{S}^{\alpha}_{x \rightarrow {t, x}}= \mathfrak{S}^{\alpha} (L^2_x \rightarrow L^2_{t, x}) ,
\ \ 
 \mathfrak{S}^{\alpha}_{t, x \rightarrow x}= \mathfrak{S}^{\alpha} (L^2_{t, x} \rightarrow L^2_x)
\end{gather*}
Here $L^2_{t,x} = L^2_t X$ for some Hilbert space $X$ and $x$ denotes the variable in $X$.

The next lemma is used in the proof of Theorem \ref{2311022111}.

\begin{lemma}[\cite{FS}, \cite{Ha}]\label{1113719}
Let $p, q, \alpha \in[1, \infty]$, $A :\mathbb{R} \rightarrow \mathcal{B}(L^2(X))$ be a strongly continuous function. Then the following are equivalent:

(1) For any $\gamma \in \mathfrak{S}^{\alpha} (L^2(X))$, 
\[\| \rho(A(t) \gamma A(t)^*)\|_{L^p_t L^q_x} \lesssim \| \gamma \|_{\mathfrak{S}^\alpha} .\]

(2) For any $f\in L^{2p'}_t L^{2q'}_x$, 
\[\|fA(t)\|_{\mathfrak{S}^{2\alpha'} _{x\rightarrow (t,x)}} \lesssim \|f\|_{L^{2p'}_t  L^{2q'}_x}.\] 
\end{lemma}
\begin{remark}

In [FS] and [Ha], the above lemma is proved for $L^2(\mathbb{R}^d)$. However, the modification of the proof in the case of $L^2(X)$ is straightforward, so we omit the proof here.
\end{remark}

The next lemma is also important in the proof of theorem \ref{2311022111}. It is the Christ-Kiselev type lemma in the Schatten spaces. See \cite{CK} for the usual Christ-Kiselev lemma.

\begin{lemma}[\cite{GK}, \cite{BS}, \cite{Ha}]\label{2311152210}

Let $-\infty \le a < b \le\infty$, $\alpha \in(1, \infty)$, $\mathcal{H}$ be an arbitrary Hilbert space. Let $K :(t,\tau) \mapsto K(t,\tau) \in \mathcal{B}(\mathcal{H})$ be a strongly continuous function. Assume that 
\[ \widetilde{T} g(t)=\int_a^b K(t, \tau) g(\tau) d\tau\]
defines a bounded operator on $L^2_t \mathcal{H}$ and $\widetilde{T} \in \mathfrak{S}^{\alpha}(L^2_t \mathcal{H})$.
Then 
\[Tg(t)=\int_a^t K(t, \tau) g(\tau) d\tau\]
 also satisfies $T \in \mathfrak{S}^{\alpha}(L^2_t \mathcal{H})$ and $\|T\|_{\mathfrak{S}^{\alpha} }  \lesssim \| \widetilde{T}\|_{\mathfrak{S}^{\alpha}}.$
\end{lemma}

\begin{remark}

In Theorem 3.1 of \cite{Ha}, it is assumed that $\mathcal{H}=L^2(\mathbb{R}^d)$. However, the proof in the case of $\mathcal{H}$ is just the same as the comments in Section 3.1 of \cite{Ha}. So we omit the proof here.
\end{remark}

\subsubsection*{\textbf{Orthonormal Strichartz estimates for the free Hamiltonian}}

Here we provide the orthonormal Strichartz estimates for the free Laplacian.

\begin{thm}[\cite{FLLS}, \cite{FS}, \cite{BHLNS}]\label{1113700}
The statements of Corollray~\ref{11161517} hold ture for the free Hamiltonian $H=-\Delta$.
\end{thm}

\begin{remark}

In the above theorem, assumptions on $\alpha$ are known to be optimal. See Theorems 1.1 and 1.2 in \cite{BHLNS} for the precise statement of the optimality.
\end{remark}

\section{Proofs of main theorems}

In this section we prove the main theorems.

\begin{proof}[Proof of Theorem~\ref{2311022111}] 

We use the duality principle by Frank-Sabin \cite{FS} and the smooth perturbation theory by Kato \cite{KY}.
Since  
\[\rho(e^{-itH} P_{ac} (H) \gamma e^{itH} P_{ac} (H))=\sum_{n=0}^\infty \nu_n |e^{-itH} P_{ac} (H) f_n|^2\] and $\| \gamma \|_{\mathfrak{S}^\alpha}=(\sum_{n=0}^\infty |\nu_n|^{\alpha})^{1/ \alpha} $ for $\gamma = \sum_{n=0}^\infty \nu_n |f_n\rangle \langle f_n|$, it suffices to show 
\[\| \rho(e^{-itH} P_{ac} (H)\gamma e^{itH} P_{ac} (H))\|_{L^p_t L^q_x} \lesssim \| \gamma \|_{\mathfrak{S}^\alpha}.\]
By Lemma \ref{1113719}, we prove that 
\[\|fe^{-itH} P_{ac} (H)\|_{\mathfrak{S}^{2\alpha'} _{x\rightarrow (t,x)}} \lesssim \|f\|_{L^{2p'}_t  L^{2q'}_x}\] holds  for any $f\in L^{2p'}_t L^{2q'}_x$.
By the Duhamel formula (\cite{BM}, \cite{M1}, \cite{MY1}, \cite{MY2}), we have
\begin{gather*} 
fU_{H} P_{ac} (H)=fU_{H_0} P_{ac} (H) -i f \Gamma_{H_0} V U_{H} P_{ac}(H), \\
U_{H}=e^{-itH} ,\ U_{H_0}=e^{-itH_0},\ \Gamma_{H_0} g(t) = \int_0^t e^{i(t-s)H_0} g(s) ds
\end{gather*}
(we only need to multiply $P_{ac}(H)$ from the right and $f$ from the left of \cite[(5.2)]{M1}).
The first term in the right hand side is estimated as
\[\|f U_{H_0} P_{ac}(H)\|_{\mathfrak{S}^{2\alpha'}_{x \rightarrow (t,x)}} \lesssim \|fe^{-itH_0}\|_{\mathfrak{S}^{2\alpha'} _{x \rightarrow (t,x)}}  \|P_{ac}(H)\|_{\mathcal{B}(\mathcal{H})}
\lesssim \|f\|_{L^{2p'}_t  L^{2q'}_x} \]
since $H_0$ satisfies the orthonormal Strichartz estimates by the assumption. Here we have used Lemma \ref{1113719} again.

Next we estimate the second term in the right hand side of the Duhamel formula as follows. First, we have
\begin{align*}
\| f \Gamma_{H_0} V U_{H} P_{ac}(H)\|_{\mathfrak{S}^{2\alpha'} _{x \rightarrow (t,x)} } 
&\lesssim \| f\Gamma_{H_0}Y^*\|_{\mathfrak{S}^{2\alpha'}_{t,x}} \|ZU_{H} P_{ac}(H)\|_{\mathcal{B}(\mathcal{H} \rightarrow L^2_{t,x})} \\
&\lesssim  \| f\Gamma_{H_0}Y^*\|_{\mathfrak{S}^{2\alpha'}_{t,x}}
.
\end{align*}
In the last line we have used the assumption that $ZP_{ac}(H)$ is $H$-smooth.
By Lemma \ref{2311152210}, we only need to estimate 
 \begin{gather*} 
\| f\widetilde{\Gamma_{H_0}}Y^*\|_{\mathfrak{S}^{2\alpha'}_{t,x}};\quad 
\widetilde{\Gamma_{H_0}} g(t)= \int_0^\infty e^{i(t-s)H_0} g(s) ds.
 \end{gather*}
Since
\begin{gather*}
f \widetilde{\Gamma_{H_0}} Y^* = fU_{H_0} \cdot  U_{H_0}^\dagger  Y^*;\quad  
U_{H_0}^\dagger g= \int_0^\infty e^{-isH_0} g(s) ds
\end{gather*}
and $U_{H_0}^\dagger$ is the formal adjoint of $U_{H_0}$, we have
\begin{align*}
\| f\widetilde{\Gamma_{H_0}}Y^*\|_{\mathfrak{S}^{2\alpha'}_{t,x}} 
&\lesssim \|fU_{H_0}\|_{\mathfrak{S}^{2\alpha'}_{x \rightarrow t,x}} \cdot \|U_{H_0}^\dagger Y^*\|_{\mathcal{B}(L^2_{t,x \rightarrow x})}
\lesssim  \|f\|_{L^{2p'}_t  L^{2q'}_x} \cdot \|YU_{H_0}\|_{\mathcal{B}(L^2_{x \rightarrow t,x})} 
\\&
\lesssim  \|f\|_{L^{2p'}_t  L^{2q'}_x}.
\end{align*}
Here in the second inequality, we have used the orthonormal Strichartz estimates for $H_0$ and Lemma \ref{1113719} and in the last line, we have used the assumption that Y is $H_0$-smooth.

Hence we have 
\[\|f U_{H} P_{ac}(H)\|_{\mathfrak{S}^{2\alpha'}_{x \rightarrow (t,x)}} \lesssim \|f\|_{L^{2p'}_t  L^{2q'}_x}.\]
By Lemma \ref{1113719} we obtain the desired estimates. 
\end{proof}

\begin{proof}[Proof of Corollary \ref{11161517}]
By Theorem \ref{1113700} and Theorem \ref{2311022111}, it suffices to show that $Y=|V|^{1/2} \sgn V$ is $-\Delta$-smooth and $Z=|V|^{1/2} P_{\ac} (H)$ is $H$-smooth. They follow from the next proposition. Note that $|V(x)| \lesssim \langle x \rangle ^{-2-\epsilon}$ holds.
\end{proof}

\begin{proposition}\label{11122216}

Under the condition of Corollary \ref{11161517}, $\langle x \rangle ^{-1-\epsilon} P_{\ac} (H)$ is $H$-smooth.
\end{proposition}

\begin{proof}

This is the same as Proposition 4.5 in \cite{RS} since $\sigma_{\ac}(H)=[0, \infty)$. So we omit the proof.
\end{proof}

\begin{proof}[Proof of Corollary \ref{ooo}]

As the previous proof,  it suffices to show that $Y=|V|^{1/2} \sgn V $ is $-\Delta$-smooth and 
$Z=|V|^{1/2} P_{\mathrm{ac}} (H)$ is $H$-smooth. However, they follow from Theorem 2.5 in \cite{BM}.
\end{proof}

Concerning the proof of Theorem \ref{1113026} (1), we prove more general results.
\begin{assump}\label{1113054}
$V \in L^1_{loc}(\R^d)$  is radial and satisfies the following
\begin{itemize}

\item $|x|^2 (x\cdot \nabla)^l V(x) \in L^{\infty}$ for $l =0, 1.$

\item There exists $\nu >0$ satisfying $|x|^2 V \ge - \frac{(d-1)^2}{4} +\nu $ $\quad$and $-|x|^2 (V+x\cdot \nabla V) \ge - \frac{(d-1)^2}{4} +\nu.$

\item $\langle (-\Delta +V)u, u \rangle \gtrsim -\|u\|^2 _2$ for $u \in C^{\infty}_0 (\R^d).$

\end{itemize}

\end{assump}

\begin{thm}\label{1113059}
Let $V$ as in Assumption \ref{1113054} and $H=-\Delta +V$. Then
 \[ \left\| \sum_{n=0}^ \infty{\nu_n|e^{-itH}P_{\rad}^{\perp} f_n|^2} \right\|_{L^p_t L^q_x} \lesssim \| \nu_n\|_{\ell^\alpha}\]
holds for $p, q, \alpha$ satisfying $(S1)$ or $(S2)$ . 
Here, $P_{\rad}$ and $P_{\rad}^{\perp}$ are from Theorem~\ref{1113026}.
\end{thm}

For the proof of Theorem \ref{1113059}, we use the next proposition (\cite{M2}).
\begin{proposition}\label{1113101}
Let $V, H$ as in Theorem \ref{1113059}. Then $|V|^{1/2} P_{\rad}^{\perp}$ is $H$-supersmooth. 
\end{proposition}

\begin{proof}[Proof of Theorem \ref {1113059}]
By repeating the proof of Theorem \ref{2311022111} substituting $P_{\rad}^{\perp}$ for $P_{\ac}(H)$, it suffices to show that $|V|^{1/2} \sgn V$ is $-\Delta$-smooth and $|V|^{1/2}P_{\rad}^{\perp}$ is $H$-smooth. The latter follows from Proposition \ref{1113101}. For the former case, since $|V(x)|^{1/2} \lesssim |x|^{-1}$ holds by Assumption \ref{1113054}, $|V|^{1/2} \sgn V \in M^{d, 2\sigma} $ holds for some $ \sigma \in ( \frac{d-1}{2}, \frac{d}{2} ] $. Then by Proposition 5.3 in \cite{BM}, $|V|^{1/2} \sgn V$ is $-\Delta$-smooth.
\end{proof}

Next we prove Theorem \ref{1113026} (2). We use the ground state representation of $H=-\Delta -(d-2)^2|x|^{-2}/4$. The ground state representation is used to prove the Strichartz estimates in \cite{Su} for non-endpoint cases and in \cite {M2} for endpoint cases. Precisely speaking, in the latter paper, weak type estimates;
 \[ \|e^{-itH} P_{\rad}u\|_{L^{2}_t L^{2^* , \infty}_x} \lesssim \|u\|_2 , \quad   2^*=\frac{2d}{d-2} \] 
 are proved. Mizutani also proved that $L^{2^* , \infty}$ is sharp in the sense that for any $q \in [1, \infty)$, there exists a radial function $\psi \in L^2 (\R^d)$ such that $e^{-itH} \psi$ does not belong to $L^2_t L^{2^* , q}_x$. In \cite{MY1}, the ground state representation for higher order or fractional operators is used to prove the weak type Strichartz estimates.

\begin{lemma}(Ground state representation)\label{1113241}
Define the unitary operator $U$ as $U: L^2_{\rad} (\R^d) \rightarrow L^2_{\rad}(\R^2) , f \mapsto (\frac{\omega_d}{\omega_2})^{1/2} |x|^{\frac{d-2}{2}}f$. Then we have
\[ Hu= -U^* \Delta_{\R^2}Uu \quad for \quad u \in P_{rad}D(H).\]
Here $U^* f =(\frac{\omega_d}{\omega_2})^{-1/2} |x|^{-\frac{d-2}{2}}f$ and $\omega_{d} =|S^{d-1}|.$
\end{lemma}
The proof of Lemma \ref{1113241} is omitted since it is proved in Proposition 3.1 in \cite{M2}. From this lemma, we have
\[e^{-itH}P_{\rad} = U^* e^{it\Delta_{\R^2}}UP_{\rad}.\]
See (3.2) in the proof of Proposition 3.1 in \cite{M2}.

We also use the following refined orthonormal Strichartz estimates.
\begin{lemma}\label{1113652}
Let $d=2$ and $p, q, \alpha$ as in $(S1)$ with $d=2$. Then the following estimate holds for all the orthonormal systems $\{f_n\}$ in $L^2 (\R^d).$
\[ \left\| \sum_{n=0}^ \infty{\nu_n|e^{it\Delta_{\R^2}}UP_{\rad}f_n|^2} \right\|_{L^p_t L^{q, p}_x} \lesssim \| \nu_n\|_{\ell^\alpha}.\] 
\end{lemma} 

\begin{remark}\label{1113706}
By almost the same way as in the proof of Lemma \ref{1113652}, we can prove the following estimate:
\[ \left\| \sum_{n=0}^ \infty{\nu_n|e^{it\Delta}f_n|^2} \right\|_{L^p_t L^{q, p}_x} \lesssim \| \nu_n\|_{\ell^\alpha}.\]
for $d, p, q, \alpha$ satisfying $(S1)$. Note that this estimate is a refinement of Theorem \ref{1113700} if $q \ge 1+\frac{2}{d}$ since $q \ge p \Leftrightarrow q \ge 1+\frac{2}{d}.$
\end{remark}

\begin{proof}[Proof of Lemma \ref{1113652}]
Since $UP_{\rad}$ is bounded on $L^2 (\R^d)$,
\begin{align}\|fe^{it\Delta_{\R^2}}UP_{\rad}\|_{\mathfrak{S}^{2\alpha'}(L^2(\R^d) \rightarrow L^2_t L^2 (\R^2))} \lesssim \|fe^{it\Delta_{\R^2}}\|_{{\mathfrak{S}^{2\alpha'}}_{x \rightarrow t, x}} \lesssim \|f\|_{L^{2p'}_t L^{2q'}_x}   \quad   \label{11162203}
\end{align}
holds. In a way similar to Lemma \ref{1113719}, we have 
\begin{align} \left\| \sum_{n=0}^ \infty{\nu_n|e^{it\Delta_{\R^2}}UP_{\rad}f_n|^2} \right\|_{L^p_t L^{q}_x} \lesssim \| \nu_n\|_{\ell^\alpha}   
 \label{11162206}
\end{align}
as follows. Set $A(t) :=e^{it\Delta_{\R^2}}UP_{\rad}: L^2 (\R^d) \rightarrow L^2 (\R^2)$ and $A_{1}: L^2 (\R^d) \rightarrow L^{\infty} _t L^2 (\R^2) , u \mapsto A(t)u(x)$. Then $A_2 :L^1_t L^2 (\R^2) \rightarrow L^2 (\R^d) , f \mapsto \int_{\R} A(s)^* f(s) ds$ is the formal adjoint of $A_1$. Now for $\gamma \in \mathfrak{S}^{\alpha} (\R^d)$, 
\begin{align*}
\| \rho (A(t) \gamma A(t)^*)\|_{L^p_t L^q_x} 
&= \sup \left\{ \left|\int_{\R \times \R^2} \rho (A(t) \gamma A(t)^*)f(t, x) dtdx\right| \mid \|f\|_{L^{p'}_t L^{q'}_x} \le1 \right\} \\
& =\sup \left\{ \left|\int_{\R} \tr (A(t) \gamma A(t)^*f(t)) dt\right| \mid \|f\|_{L^{p'}_t L^{q'}_x} \le1 \right\} \\
& =\sup \left\{ \left|\int_{\R} \tr ( \gamma A(t)^*f(t) A(t)) dt\right| \mid \|f\|_{L^{p'}_t L^{q'}_x} \le1 \right\} \\
& =\sup \left\{ \left| \tr ( \gamma A_2 f(t) A_1) \right| \mid \|f\|_{L^{p'}_t L^{q'}_x} \le1 \right\} \\
& \lesssim \| \gamma \|_{\mathfrak{S}^{\alpha}} \cdot \sup \left\{ \|A_2 f(t) A_1)\|_{\mathfrak{S}^{\alpha'}} \mid \|f\|_{L^{p'}_t L^{q'}_x} \le1 \right\} \\
& \lesssim \| \gamma \|_{\mathfrak{S}^{\alpha}} \cdot \sup \left\{ \| f(t)^{\frac{1}{2}} A_1)\|_{\mathfrak{S}^{2\alpha'} (L^2 (\R^d) \rightarrow L^2_t L^2 (\R^2))} \mid \|f\|_{L^{p'}_t L^{q'}_x} \le1 \right\} \\
& \lesssim \| \gamma \|_{\mathfrak{S}^{\alpha}} .
\end{align*}
Here in the third line, we have used $\tr(ST)=\tr(TS)$ for $S: \mathfrak{H} \rightarrow \mathfrak{K}$ and $T: \mathfrak{K} \rightarrow \mathfrak{H}$, in the fifth line we have used H\"{o}lder's inequality for Schatten spaces, in the sixth line we have used the fact that $A_2$ is the formal adjoint of $A_1$
and in the last line (\ref{11162203}) is used. Thererfore by taking $\gamma = \sum_{n=0}^{\infty} \nu_n |f_n \rangle \langle f_n |$, we have (\ref{11162206}).

Next we use the following results for real interpolation.
\begin{itemize}

\item $(\ell^{p_0}, \ell^{p_1})_{\theta, q} = \ell^{p, q}$. Here $1/p = (1-\theta)/p_0 + \theta/p_1$

\item $(L^{p_0}_t L^{q_0}_x , L^{p_1}_t L^{q_1}_x)_{\theta, p} = L^p _t L^{q, p}_x$. Here $1/p = (1-
\theta)/p_0 + \theta/p_1$ and $1/q = (1-\theta)/q_0 + \theta/q_1$

\end{itemize}
By interpolating (\ref{11162206}) for $(p_j , q_j)$, $j=0,1$, we have 
\[ \left\| \sum_{n=0}^ \infty{\nu_n|e^{it\Delta_{\R^2}}UP_{\rad}f_n|^2} \right\|_{L^p_t L^{q, p}_x} \lesssim \| \nu_n\|_{\ell^{\frac{2q}{q+1}, p}} \]
since $\alpha =\frac{2q}{q+1}$. Since $\frac{2q}{q+1} \le p \Leftrightarrow q \le 3$ and this is true under $(S1)$, we have the desired estimate.
\end{proof}

\begin{proof}[Proof of Theorem \ref{1113026} (2)]
By changing into polar coordinates, we obtain
\begin{align*}
 \left\| \sum_{n=0}^ \infty{\nu_n|e^{-itH}P_{\rad} f_n|^2} \right\|_{L^p_t L^q_x} & \lesssim  \left\| r^{\frac{d-2}{q}}\sum_{n=0}^ \infty{\nu_n|e^{-itH}P_{\rad} f_n|^2} \right\|_{L^p_t L^q (\R^2_x)} \\
& = \left\| r^{\frac{d-2}{q} -(d-2)}\sum_{n=0}^ \infty{\nu_n|e^{it\Delta_{\R^2}}UP_{\rad} f_n|^2} \right\|_{L^p_t L^q (\R^2_x)}.
\end{align*}
Set $\tilde{q} =p'$ and $1/r = \frac{d-2}{2} (1-1/q)$. Then by H\"{o}lder's inequality, we have
\begin{align*}
\left\| r^{\frac{d-2}{q} -(d-2)}\sum_{n=0}^ \infty{\nu_n|e^{it\Delta_{\R^2}}UP_{\rad} f_n|^2} \right\|_{L^q (\R^2_x)} & \lesssim \left\|\sum_{n=0}^ \infty{\nu_n|e^{it\Delta_{\R^2}}UP_{\rad} f_n|^2} \right\|_{L^{\tilde{q}, q} (\R^2_x)} \cdot \left\| r^{\frac{d-2}{q} -(d-2)} \right\|_{L^{r, \infty} (\R^2)} \\
& \lesssim \left\|\sum_{n=0}^ \infty{\nu_n|e^{it\Delta_{\R^2}}UP_{\rad} f_n|^2} \right\|_{L^{\tilde{q}, q} (\R^2_x)}.
\end{align*}
From now on, we assume $p \le q \Leftrightarrow 1+ \frac{2}{d} \le q < 1+ \frac{2}{d-1}$ holds. Under the assumption in Theorem \ref{1113026} (2), we also have $1 \le \tilde{q} <3 \Leftrightarrow q< \frac{3d}{3d-4}$. Note that $\frac{3d}{3d-4} < 1+\frac{2}{d-1} \Leftrightarrow d \ge 3$ and $1+ \frac{2}{d} < \frac{3d}{3d-4} \Leftrightarrow d \le 3$ hold. Hence we have
\begin{align*}
\left\|\sum_{n=0}^ \infty{\nu_n|e^{it\Delta_{\R^2}}UP_{\rad} f_n|^2} \right\|_{L^p_t L^{\tilde{q}, q} (\R^2_x)} & \lesssim \left\|\sum_{n=0}^ \infty{\nu_n|e^{it\Delta_{\R^2}}UP_{\rad} f_n|^2} \right\|_{L^p_t L^{\tilde{q}, p} (\R^2_x)} \\
& \lesssim \|\nu\|_{\ell^{\frac{2\tilde{q}}{\tilde{q}+1}}}.
\end{align*}
In the last inequality, we have used Lemma \ref{1113652} since $2/p + 2/\tilde{q} =2$ holds.
Under our assumption, note that $\frac{2\tilde{q}}{\tilde{q}+1} \le \frac{2{q}}{{q}+1}$ holds. Therefore we have
\[ \left\| \sum_{n=0}^ \infty{\nu_n|e^{-itH}P_{\rad} f_n|^2} \right\|_{L^p_t L^q_x} \lesssim  \|\nu\|_{\ell^{\frac{2{q}}{{q}+1}}}.\]
Finally by interpolating this estimate and the trivial estimate:
\[ \left\| \sum_{n=0}^ \infty{\nu_n|e^{-itH}P_{\rad} f_n|^2} \right\|_{L^1_t L^{\infty}_x} \lesssim  \|\nu\|_{\ell^1},\]
we have the desired estimates for all $d, p, q$ satisfying the assumption in Theorem \ref{1113026} (2).
\end{proof}

For the proof of Theorem \ref{11102255}, we use the following proposition (\cite{EGS}).

\begin{proposition}[\cite{EGS}]\label{11122214}
Let $d, A, V, H$ be as in Theorem \ref{11102255}. Then we have the following smoothing estimates:
\begin{align}
& \| \langle x \rangle^{-\sigma} |D|^{1/2} e^{-itH}P_{\ac}(H)u\|_{L^2_tL^2_x} \lesssim \|u\|_2 \quad for \quad \sigma>1/2 \label{11162226} \\
& \| \langle x \rangle^{-\sigma} \langle D \rangle^{1/2} e^{-itH}P_{\ac}(H)u\|_{L^2_tL^2_x} \lesssim \|u\|_2 \quad for \quad \sigma>1 \label{11162227}.
\end{align}

\end{proposition}

\begin{proof}[Proof of Theorem \ref{11102255}]
The proof follows the same line as \cite{EGS}. Since the case of $H = (D +A(x))^2 +V(x)$ is just the same as the case of $H = -\Delta +A(x)\cdot D + D\cdot A(x) + V(x)$, we only consider $H = -\Delta +A(x)\cdot D + D\cdot A(x) + V(x)$. First, we decompose 
\[H= - \Delta + \sum_{j=1}^{3} Y_j^*Z_j .\] Here 
\begin{gather*} 
\ Y_3 =|V|^{1/2}\sgn V, \quad  Z_3 = |V|^{1/2}, \quad Y_1^* = Aw^{-1}D|D|^{-1/2}, \\
\ Z_1 =|D|^{1/2}w, \quad  Y_2 =Z_1, \quad Z_2 =Y_1 \quad and \quad  w= \langle x \rangle^{- \tau} \quad for \quad \tau \in(1/2, 1/2 +\epsilon').
\end{gather*}
By Theorem \ref{2311022111}, it suffices to show that $Y_j$ are $-\Delta$-smooth and $Z_j P_{\ac}(H)$ are $H$-smooth. 

(1) $H_0$-smoothness of $Y_j$

\noindent First, $H_0$-smoothness of $Y_3$ follows from Proposition \ref{11122216} for $V=0$. For $Y_2$, we decompose
\[ |D|^{1/2}w= (|D|^{1/2}w|D|^{-1/2}w^{-1+\eta}) \cdot w^{1-\eta}|D|^{1/2}\]
for sufficiently small $\eta$. The first factor is bounded on $L^2$ by Lemma 6.2 in \cite{EGS}. The second factor is $H_0$-smooth by Proposition \ref{11122214} (\ref{11162226}) for $A=V=0$. Therefore $Y_2$
is $H_0$-smooth. For $Y_1$, we decompose
\[ Y_1 =\frac{D}{|D|}\cdot |D|^{1/2}w^{-1}Aw^{-1}|D|^{-1/2}\cdot Y_2.\]
The first factor is bounded by the H\"{o}rmander-Mikhlin theorem. For the second factor, by the fractional Leibniz rule, we have 
\[ \| |D|^{1/2}w^{-2}A|D|^{-1/2}u\|_2 \lesssim \|w^{-2}A\|_{\infty} \|u\|_2 + \||D|^{1/2} w^{-2}A\|_{2d}\||D|^{-1/2}u\|_{\frac{2d}{d-1}} \lesssim \|u\|_2.\]
In the last inequality, we have used the Sobolev embedding and $w^{-2}A \in \dot{W}^{1/2, 2d}$ by the assumption. By the $H_0$-smoothness of $Y_2$, $Y_1$ is also $H_0$-smooth.

(2) $H$-smoothness of $Z_j P_{\ac}(H)$

\noindent $H$-smoothness of $Z_3 P_{ac}(H)$ follows from Proposition \ref{11122214} (\ref{11162227}) and the boundedness of $\langle x \rangle^{-\sigma} \langle D \rangle^{-1/2} \langle x \rangle^{\sigma}$ for $\sigma >1$. For $Z_1 P_{\ac}(H)$, we decompose
\[ Z_1 P_{\ac}(H) = (|D|^{1/2}w|D|^{-1/2} \langle x \rangle^{\sigma})\cdot \langle x \rangle^{-\sigma}|D|^{1/2} P_{\ac}(H)\]
taking $\sigma > \frac{1}{2}$ sufficiently close to $\frac{1}{2}$. The first factor is bounded by Lemma 6.2 in \cite{EGS} and the second factor is $H$-smooth by Proposition \ref{11122214} (\ref{11162226}). Hence $Z_1 P_{ac}(H)$ is $H$-smooth. For $Z_2 P_{\ac}(H)$, we have
\[ Z_2 P_{\ac}(H) = \frac{D}{|D|} \cdot (|D|^{1/2}Aw^{-1} \langle x \rangle^{1+\epsilon' -\tau}|D|^{-1/2})\cdot (|D|^{1/2} \langle x \rangle^{\tau - (1+ \epsilon')} |D|^{-1/2} \langle x \rangle^{\sigma}) \cdot \langle x \rangle^{-\sigma} |D|^{1/2}P_{\ac}(H).\]
Since the first, second and third factors are bounded by the same reason and the last factor is $H$-smooth by Proposition \ref {11122214} (\ref{11162226}), $Z_2 P_{\ac}(H)$ is $H$-smooth.
\end{proof}

\begin{proof}[Proof of Corollary \ref{jjj}]
The proof is similar to Corollary 9 in \cite{FS}.
By the Littlewood-Paley theorem for $H$ (Proposition 2.9 in \cite{M1}), we have
\begin{align*} 
\|e^{-itH}P_{\ac}(H)u\|_{L^{2p}_t L^{2q}_x} ^2
&\lesssim \left\|\sum_{n \in \Z} |e^{-itH}P_{\ac}(H) \phi_n (H)u|^2\right\|_{L^p_t  L^q_x}\\
&\lesssim \left\|\sum_{n \in \Z} |e^{-itH}P_{\ac}(H) \phi_{2n} (H)u|^2\right\|_{L^p_t  L^q_x} + \left\|\sum_{n \in \Z} |e^{-itH}P_{\ac}(H) \phi_{2n+1} (H)u|^2\right\|_{L^p_t  L^q_x} \\
&= \| \rho (e^{-itH}P_{\ac}(H) \gamma_e e^{itH}P_{\ac}(H))\|_{L^p_t L^q_x} + \| \rho (e^{-itH}P_{\ac}(H) \gamma_o e^{itH}P_{\ac}(H))\|_{L^p_t L^q_x}  .   
\end{align*}
Here $\phi_n$ is the function appearing in the definition of the Besov spaces and $\gamma_e$ and $\gamma_o$ are defined as follows:
\[
\gamma_e = \sum_{n \in \Z} |\phi_{2n}(H)u \rangle \langle \phi_{2n}(H)u|,
\quad 
\gamma_o = \sum_{n \in \Z} |\phi_{2n+1}(H)u \rangle \langle \phi_{2n+1}(H)u|.
\]
Then by Corollary \ref{11161517}, Corollary \ref{ooo}, we have
\begin{align*}
\| \rho (e^{-itH}P_{\ac}(H) \gamma_e e^{itH}P_{\ac}(H))\|_{L^p_t L^q_x} &+ \| \rho (e^{-itH}P_{\ac}(H) \gamma_o e^{itH}P_{\ac}(H))\|_{L^p_t L^q_x} \\
&\lesssim \|\gamma_e\|_{\mathfrak{S}^{\alpha}} + \|\gamma_o\|_{\mathfrak{S}^{\alpha}} \\
&= \| \|\phi_{2n}(H) u\|^2_2\|_{\ell^{\alpha}} + \| \|\phi_{2n+1}(H) u\|^2_2\|_{\ell^{\alpha}} \\
&\lesssim \| \|\phi_{n}(H) u\|^2_2\|_{\ell^{\alpha}}.
\end{align*} 
Hence
\[  \|e^{-itH}P_{\ac}(H)u\|_{L^{2p}_t L^{2q}_x} \lesssim \| \|\phi_{n}(H) u\|_2\|_{\ell^{2\alpha}} .\]
The right hand side is equivalent to the Besov norm $\| \cdot \|_{\dot{B}^0_{2,\frac{4q}{q+1}}}$ by Proposition 3.5 in \cite{IMT} since $2\alpha = \frac{4q}{q+1}$.
So we have the desired estimates.
\end{proof}
Before proving Corollary \ref{11191637}, we need some lemmas. First, we give the definition of the Kato class.
\begin{defn}\label{11191957}
Let $d \ge 3$, $V: \R^d \rightarrow \R$. Then $V \in K_{d}$ iff $\lim_{r \rightarrow 0} \sup_{x \in \R^d} \int_{|x-y| <r} \frac{|V(y)|}{|x-y|^{d-2}} dy =0$ holds. $K_d$ is called the Kato class. The Kato norm is defined by $\|V\|_{K_d} := \sup_{x \in \R^d} \int_{|x-y| <1} \frac{|V(y)|}{|x-y|^{d-2}} dy$.
\end{defn}
\noindent To prove the Littlewood-Paley theorem for the magnetic Schr\"odinger operator, we use the following estimate for spectral multipliers.
\begin{thm}[\cite{CD}]\label{11192048}
Let $H$ be a nonnegative self-adjoint operator on $L^2 (\R^d)$. Assume the integral kernel of $e^{-tH}$ satisfies the Gaussian upper bound:
\[ |e^{-tH} (x, y)| \lesssim t^{-d/2} e^{-|x-y|^2 /At}\]
for some $A>0$. Let $\phi \in C^{\infty} _0, \supp \phi \subset (1/2, 2)$ be as in the definition of the Besov space. If $g: \R_{>0} \rightarrow \C$ satisfy $\sup_{t>0} \| \phi (\cdot)g(t \cdot)\|_{H^{s}} < \infty$ for some $s> \frac{d+1}{2}$, 
\[\|g(\sqrt{H})u\|_{p} \lesssim (p + \frac{1}{p-1})\|u\|_{p}\]
holds for $p \in (1, \infty)$.
\end{thm}
Concerning the Gaussian estimates for the magnetic Schr\"odinger operator, the following are known.
\begin{proposition}[\cite{CD}]\label{11192049}
Let $H = (D+A(x))^2 +V$, $A \in L^2 _{loc} (\R^d)$, $V_{+} \in K_d$, $\|V_{-}\|_{K_d} < \frac{\pi^{d/2}}{\Gamma(d/2 -1)}$. Then $H \ge 0$ and $|e^{-tH} (x, y)| \lesssim t^{-d/2} e^{-|x-y|^2 /8t}$ holds.
\end{proposition}
Using these estimates, we prove the Littlewood-Paley theorem for $H$. If $A=0$, the following is proved in \cite{M1} under more general assumptions.
\begin{proposition}\label{11192111}
Assume $H$ as in Proposition \ref{11192049} and $\sigma_{p} (H) = \emptyset$. Then we have
\[ \|u\|_{p} \approx \|\|\{\phi_j (\sqrt{H}) u\}\|_{\ell^2}\|_{p}\]
for all $p \in (1, \infty)$. Here $\phi _j$ is the homogeneous Littlewood-Paley decomposition.
\end{proposition}

\begin{proof}
We follow the argument in \cite{M1}. Set
\[Su(x) = \|\{\phi_j (\sqrt{H}) u\}\|_{\ell^2}, \quad S_{\pm} u=\left(\sum_{\pm j \ge0} |\phi_j (\sqrt{H}) u|^2\right)^{1/2}.\]
Since zero is not an eigenvalue,
\[\langle u, v \rangle = \sum_{j, k \in \Z} \langle \phi_j (\sqrt{H}) u, \phi_k (\sqrt{H}) v \rangle\]
holds. By the almost orthogonality of $\phi_j (\sqrt{H})$ and H\"older's inequality, we have
\begin{align*}
|\langle u, v \rangle| \le \sum_{j, k \in \Z} \int |\phi_j (\sqrt{H}) u \phi_k (\sqrt{H}) v| dx &= \int \sum_{j, k \in \Z, |j-k| \le 2} |\phi_j (\sqrt{H}) u \phi_k (\sqrt{H}) v| dx \\
& \le 5 \int Su(x) Sv(x) \le 5\|Su\|_p \|Sv\|_{p'}.
\end{align*}
If we prove $\|Su\|_p \lesssim \|u\|_p$ for all $p \in (1, \infty)$, we obtain $\|u\|_p \lesssim \|Su\|_p$ by the above inequality and the duality argument. First we consider $S_+$. Let $r_j (t)$ be the Rademacher functions on $[0, 1]$, i.e. $r_0 (t) =1$ for $t \in [0, 1/2]$, $r_0 (t)=-1$ for $t \in (1/2, 1)$, $r_j (t) =r_0 (2^j t)$ for $j \ge 1$. Set
\[ m_t (s) = \sum_{j \ge 0} r_j (t) \phi (\frac{s}{2^j}).\]
By the support property of $\phi_j$, we have $|\partial^k _s m_t (s)| \lesssim s^{-k}$. Therefore $m_t$ satisfies the condition on $g$ in Theorem \ref{11192048} since $\supp \phi$ is away from zero. Hence
\[\|m_t (\sqrt{H})u\|_p \lesssim \|u\|_p\]
holds for $p \in (1, \infty)$. Now we use Khintchin's inequality:
\[ \{a_m\} \in \ell^2, F(t):= \sum_{m \in \Z_{\ge0}} a_m r_m (t) \Rightarrow \|F\|_{L^p _t} \approx \|\{a_m\}\|_{\ell^2}\]
for all $p \in (1, \infty)$. Then we have
\begin{align*}
\|S_{+} u\|_p \approx \|\|m_t (\sqrt{H})u\|_{L^p _t}\|_p = \|m_t (\sqrt{H})u\|_{L^p ([0, 1]: L^p _x)} \lesssim \|u\|_p
\end{align*}
for all $p \in (1, \infty)$. For $S_{-}$, Set $m_t (s) = \sum_{j \ge 0} r_j (t) \phi (2^j s)$. Then $m_t \in C^{\infty} _0$ and the above argument can be also applied to this case. Hence we are done.
\end{proof}
We need one more lemma, which provides a relation between the usual Besov space and the Besov space associated to the Schr\"odinger operator $H$.
\begin{lemma}\label{1120007}
Assume $H =(D+A(x))^2 +V \ge0$, $|A(x)| \lesssim \langle x \rangle^{-1-\epsilon}, |V(x)| \lesssim \langle x \rangle^{-1-\epsilon}$. Furthermore we assume there exists $s_0, s_1$ such that $s_0 < 0<s_1$, $H^{s_j} H^{-s_j} _0$ is bounded on $L^2 (\R^d)$ for $j=0, 1$. Then we have
\[\|u\|_{\dot{B}^0 _{2, q} (H)} \lesssim \|u\|_{\dot{B}^0 _{2, q}}\]
for all $q \in [1, \infty]$. Here $\dot{B}^0 _{2, q} (H)$ is the homogeneous Besov space associated to $H$:
\[\|u\|_{\dot{B}^0 _{2, q} (H)}:= \|\{\|\phi_j (\sqrt{H})u\|_2\}\|_{\ell^q}.\]
\end{lemma}
To prove Lemma \ref{1120007}, we use the existence of the wave operator for $H, H_0$. In \cite{IS}, the existence and the asymptotic completeness of the wave operator:
\[W_{\pm} = s-\lim_{t \rightarrow \pm \infty} e^{itH} e^{-itH_0}P_{\ac} (H_0)\]
are proved under more general assumptions on perturbations.
\begin{proof}[Proof of Lemma \ref{1120007}]
By the intertwining property of wave operators, we have
\[f(H) = W_{\pm} f(H_0) W^*_{\pm}\]
for $f \in L^2 _{loc} (\R).$ Here we have used $\sigma (H) = \sigma_{\ac} (H) =[0, \infty)$ since $H \ge0$ ($\sigma_{\ac} (H) =[0, \infty)$ is proved in \cite{IS}). Then
\[\|\phi_j (\sqrt{H})u\|_2 = \|W_{\pm} \phi_j (|D|) W^*_{\pm} u\|_2 \le \|\phi_j (|D|) W^*_{\pm} u\|_2\]
holds. Hence we obtain
\[\|u\|_{\dot{B}^0 _{2, q} (H)} \lesssim \|W^*_{\pm}u\|_{\dot{B}^0 _{2, q}}.\]
By the assumption, we have
\[\|H^{s_j} _0 W^*_{\pm} u\|_2 = \|W^* _{\pm} H^{s_j} u\|_2 \le \|H^{s_j} u\|_2 \lesssim \|H^{s_j} _0 u\|_2 .\]
Hence $W^* _{\pm}$ is a bounded operator on $\dot{H}^{2s_j}$. Since $(\dot{H}^{s_0}, \dot{H}^{s_1})_{\theta, q} = \dot{B}^s _{2, q}$, $s=(1-\theta)s_0 + \theta s_1$, $W^* _{\pm}$ is also bounded on $\dot{B}^0 _{2, q}$. Therefore we obtain
\[\|u\|_{\dot{B}^0 _{2, q} (H)} \lesssim \|u\|_{\dot{B}^0 _{2, q}}.\]
\end{proof}

\begin{proof}[Proof of Corollary \ref{11191637}]
Since $\langle Hu, u \rangle \approx \langle H_{0} u, u \rangle$, $H^{\pm 1/2} H^{\mp 1/2} _0$ is bounded on $L^2$. Hence the assertion in Lemma {\ref{1120007}} holds. By the assumption, Proposition \ref{11192111} also holds. Hence by repeating the proof of Corollary \ref{jjj}, we have the desired estimates.
\end{proof}

Finally we prove Theorem \ref{11102301}. We use the following lemma. This estimate was first proved in \cite{FLLS} when $H=-\Delta$.
\begin{lemma}\label{11131023}
Let $d, p, q$ be as in $(S1)$ , $H$ be as in Theorem \ref{11102301} and $\gamma$ be a solution to 
\[
\left\{
\begin{array}{l}
i\partial_t \gamma=[H, \gamma(t)] + R(t) \\
\gamma(0)=\gamma_0 \in \mathfrak{S}^{\frac{2q}{q+1}}
\end{array}
\right.
\]
i.e. 
\[ \gamma(t) = e^{-itH} \gamma_{0} e^{itH} -i \int_{0}^{t} e^{-i(t-s)H} R(s)e^{i(t-s)H} ds.\]
Then we have
\[ \| \rho_{\gamma(t)}\|_{L^p_t L^q_x} \lesssim \|\gamma_0\|_{\mathfrak{S}^{\frac{2q}{q+1}}} + \left\|\int_{\R} e^{isH} |R(s)| e^{-isH} ds \right\|_{\mathfrak{S}^{\frac{2q}{q+1}}}.\]
Here $|T| = (T^* T)^{1/2}$.
\end{lemma}
\begin{proof}
By the orthonormal Strichartz estimates for $H$, we have
\begin{align*}
\| \rho_{\gamma(t)}\|_{L^p_t L^q_x} &\lesssim \| \rho (e^{-itH} \gamma_{0} e^{itH}) \|_{L^p_t L^q_x} + \left\|\rho \left(\int_{0}^{t} e^{-i(t-s)H} R(s)e^{i(t-s)H} ds \right) \right\|_{L^p_t L^q_x} \\
& \lesssim \|\gamma_0\|_{\mathfrak{S}^{\frac{2q}{q+1}}} + \left\|\rho \left(\int_{0}^{t} e^{-i(t-s)H} R(s)e^{i(t-s)H} ds \right) \right\|_{L^p_t L^q_x}.
\end{align*}
Set $\tilde{\gamma}(t) = \int_{0}^{t} e^{-i(t-s)H} R(s)e^{i(t-s)H} ds$. For $f \in C^{\infty}_{0} (\R \times \R^d)$, we have
\begin{align*}
\left| \int_{0}^{\infty} \int_{\R^d} f(t, x) \rho_{\tilde{\gamma}(t)}(x) dxdt \right| &= \left| \int_{0}^{\infty} \tr (f(t) \tilde{\gamma}(t)) dt \right| \\
& = \left| \int_{0}^{\infty} \int_{0}^{t} \tr (f(t) e^{-i(t-s)H} R(s)e^{i(t-s)H}) dsdt \right| \\
& = \left| \int_{0}^{\infty} \int_{0}^{t} \tr ( e^{itH}f(t) e^{-itH} e^{isH}R(s)e^{-isH}) dsdt \right| \\
& \lesssim  \int_{0}^{\infty} \int_{0}^{t} \tr ( e^{itH}|f(t)| e^{-itH} e^{isH}|R(s)|e^{-isH}) dsdt  \\
& \le  \int_{0}^{\infty} \int_{0}^{\infty} \tr ( e^{itH}|f(t)| e^{-itH} e^{isH}|R(s)|e^{-isH}) dsdt \\
& = \tr \left( \left(\int_{0}^{\infty} e^{itH}|f(t)| e^{-itH} dt \right) \left (\int_{0}^{\infty} e^{isH}|R(s)|e^{-isH} ds\right)\right).
\end{align*}
In the fourth line we have used $|\tr(AB)| \le \tr(|A||B|)$. For $\hat{\gamma} \in \mathfrak{S}^{\frac{2q}{q+1}}$, set $\hat{\gamma}(t) = e^{-itH} \hat{\gamma} e^{itH}$. Then
\begin{align*}
\left| \tr \left( \left( \int_{0}^{\infty} e^{itH}|f(t)| e^{-itH} dt \right) \hat{\gamma} \right)\right| &= \left| \int_{0}^{\infty} \tr (|f(t)| e^{-itH} \hat{\gamma} e^{itH}) dt \right| \\
& = \left| \int_{0}^{\infty} \int_{\R^d} |f(t)| \rho_{\hat{\gamma}(t)} dxdt \right| \\
& \lesssim \|f\|_{L^{p'}_t L^{q'}_x} \cdot \|\rho_{\hat{\gamma}(t)}\|_{L^p_t L^q_x} \\
& \lesssim \|f\|_{L^{p'}_t L^{q'}_x} \|\hat{\gamma}\|_{\mathfrak{S}^{\frac{2q}{q+1}}}.
\end{align*}
In the last line the orthonormal Strichartz estimates are used. Then by the duality of the Schtten spaces, we have
\[ \left\|\int_{0}^{\infty} e^{itH}|f(t)| e^{-itH} dt \right\|_{\mathfrak{S}^{2q'}} \lesssim \|f\|_{L^{p'}_t L^{q'}_x}.\]
Hence
\[\left| \int_{0}^{\infty} \int_{\R^d} f(t, x) \rho_{\tilde{\gamma}(t)}(x) dxdt \right| \lesssim \|f\|_{L^{p'}_t L^{q'}_x} \cdot \left\| \int_{0}^{\infty} e^{isH}|R(s)|e^{-isH} ds \right\|_{\mathfrak{S}^{\frac{2q}{q+1}}} .\]
By the same way we can estimate the integral on $(-\infty, 0] \times \R^d$. So by the duality argument, we have
\[\|\rho_{\tilde{\gamma}(t)}\|_{L^p_t L^q_x} \lesssim \left\| \int_{-\infty}^{\infty} e^{isH}|R(s)|e^{-isH} ds \right\|_{\mathfrak{S}^{\frac{2q}{q+1}}}.\]
\end{proof}
We also use the following lemma for the existence of the propagator to the linear Schr\"odinger equation with time-dependent potentials. 
\begin{lemma}\label{1114032}
Assume $H$ is as in Theorem \ref{11102301}. Let $V(t, x) = V_1 (t, x) +V_2 (t, x)$ satisfy $V_1 \in L^{\alpha}_T L^p _x$ and $V_2 \in L^{\beta} _T L^{\infty} _x$ for some $p \ge 1, \alpha \ge 1, \beta >1, 0 \le 1/\alpha < 1- \frac{d}{2p}$. Here $L^q_T L^r _x$ denotes $L^q ([-T, T]; L^r _x)$ for $T \in (0, \infty]$. Then the integral equation:
\[ u(t) = e^{-i(t-s)H} u_0 -i \int_{s}^{t} e^{-i(t-s')H} V(s')u(s') ds' , \quad  t \in [-T, T]\]
has a unique solution $u \in C_T L^2 _x \cap L^{\theta} _{loc, T} L^q _x$ for $q= \frac{2p}{p-1}, \theta = \frac{4p}{d}$. Note that $(\theta, q)$ is an admissible pair. Furthermore the mass conservation law $\|u(t)\|_2 = \|u_0\|_2$ holds.
\end{lemma}

\begin{remark}
The above integral equation is a formulation of the Schr\"odinger equation:
\[
\left\{
\begin{array}{l}
i\partial_t u= Hu(t) + V(t)u(t) \\
u(s)=u_0 \in L^2_x
\end{array}
\right.
.\]
When $H= -\Delta$, the above result was proved in \cite{Y}. Concerning the formulation, it is also possible to consider the equation
\[ u(t) = e^{i(t-s)\Delta} u_0 -i \int_{s}^{t} e^{i(t-s')\Delta} (V(s') + \tilde{V})u(s') ds' , \quad  t \in [-T, T]\]
when $H= -\Delta + \tilde{V}$. However it seems difficult to solve this equation, for example, when $\tilde{V}$ is an inverse square type potential since its singularity at origin is strong. Also note that even if these two integral equations have solutions for the same $u_0$, we do not know whether these two solutions are the same or not. This is because these solutions are not differentiable generally. As for the differentiability of these solutions, Yajima gave sufficient conditions in \cite{Y}.
\end{remark}
From the above lemma we have the following.
\begin{corollary}\label{1114035}
Under the assumptions in Lemma \ref{1114032}, there exists a family of unitary operators $\{U(t, s)\}$ such that
\begin{itemize}

\item $U(t, s)$ is strongly continuous with respect to $(t, s)$.

\item $U(t, s)U(s, r) = U(t, r)$ holds.

\item $U(t, s)u_0$ is a unique solution to the integral equation in Lemma \ref{1114032}.

\end{itemize}

\end{corollary}

\begin{proof}[Proof of Lemma \ref{1114032} and Corollary \ref{1114035}]
We follow the argument in \cite{Y}. Let $q, \theta$ be as in the statement. Set
\begin{align*}
&\mathcal{X}(a) = \mathcal{X}(a, q) = C(I; L^2 _x) \cap L^{\theta} (I; L^q _x) \\
&\mathcal{X}^* (a) = \mathcal{X}^* (a, q) = L^1 (I; L^2 _x) + L^{\theta'} (I; L^{q'} _x) \\
&\mathcal{M} (a) = L^{\alpha} (I; L^p _x) + L^{\beta} (I; L^{\infty} _x)
\end{align*}
for $0< a < 1/2$ and $I = [-a, a]$. Then we have 
\[\|Vu\|_{\mathcal{X}^*(a)} \le (2a)^{\gamma} \|V\|_{\mathcal{M}(a)} \|u\|_{\mathcal{X}(a)}\]
for $\gamma = min (1-\frac{1}{\beta}, 1- \frac{d}{2p} -\frac{1}{\alpha})$. See Lemma \ref{2311152210} in \cite{Y} for its proof. By this estimate and the Strichartz estimates for $H$, we have
\begin{align*}
&\|Qu\|_{\mathcal{X}(a)} \lesssim \|Vu\|_{\mathcal{X}^* (a)} \lesssim a^{\gamma} \|V\|_{\mathcal{M}(a)} \|u\|_{\mathcal{X}(a)} \\
& Qu(t) = \int_{0}^{t} e^{-i(t-s)H} V(s)u(s) ds . 
\end{align*}
Hence if we take $a$ sufficiently small, we have
\begin{align*}
&\Phi : \mathcal{X}(a) \rightarrow \mathcal{X}(a), u \mapsto e^{-itH} u_0 -iQu \\
& \|\Phi (u-v)\|_{\mathcal{X}(a)} \le \frac{1}{2} \|u-v\|_{\mathcal{X}(a)} .
\end{align*}
Then $\Phi$ is a contraction mapping on $\mathcal{X}(a)$ and our integral equation for $s=0$ has a unique solution $u(t) = (1+ iQ)^{-1} u_0 \in \mathcal{X}(a)$. By taking $u_0 = e^{isH} u_0$ and $V = V(t+s)$, we have a unique solution $u(t)$ to the integral equation with initial data $u(s) = u_0$ on $[s-a, s+a]$. Then $U(t, s)u_0 := u(t)$ is a linear operator on $L^2 _x$ by construction. By glueing these $U(t, s)$, we have a unique solution on $[-T, T]$ and a family of linear operators $\{U(t, s)\}$ for $t, s \in [-T, T]$. Hence it suffices to show that $\|U(t, s)u_0\|_2 = \|u_0\|_2$ holds. Instead of the regularizing technique in \cite{Y}, we use the method in \cite{O}. We assume $s=0$ for the sake of simplicity. Then
\begin{align*}
\|U(t, 0)u_0\|^2 _2 &= \left\|e^{-itH} u_0 -i \int_{0}^{t} e^{-i(t-s)H} V(s)u(s) ds \right\|^2_2 \\
& = \|u_0\|^2 _2 + \left\| \int_{0}^{t} e^{-i(t-s)H} V(s)u(s) ds \right\|^2 _2 -2 \im \int_{0}^{t} \langle e^{-isH} u_0 , V(s)u(s) \rangle ds .
\end{align*}
Note that the third term is well-defined since $e^{-isH} u_0, u(s) \in C_{T} L^2 \cap L^{\theta} _{loc, T} L^{q} _x$ ensures that 
\begin{align*}
& \int_{0}^{t} |\langle e^{-isH} u_0 , V(s)u(s) \rangle| ds \\
& \lesssim _{t} \|V_1\|_{L^{\alpha} _{[0, t]} L^p _x} \|e^{-itH} u_0\|_{L^{\theta} _T L^q _x} \|u\|_{L^{\theta} _{[0, t]} L^q _x} + \int_{0}^{t} \|V_2 (s)\|_{\infty} \|u_0\|_2 \|u\|_{L^{\infty} _{T} L^2 _x} ds \\
& < {\infty}
\end{align*}
by the relation: $1/p + 1/q + 1/q =1, 1/\alpha + 1/\theta + 1/\theta <1$. We transform the second term as
\begin{align*}
\left\| \int_{0}^{t} e^{-i(t-s)H} V(s)u(s) ds \right\|^2 _2 &= \re \int_{0}^{t} \int_{0}^{t} \langle V(s)u(s) , e^{-i(s-s')H} V(s')u(s') \rangle ds' ds \\
& = 2 \re \int_{0}^{t} \int_{0}^{s} \langle V(s)u(s) , e^{-i(s-s')H} V(s')u(s') \rangle ds' ds \\
& = -2 \im \int_{0}^{t} \langle V(s)u(s) , u(s) + i\int_{0}^{s} e^{-i(s-s')H} V(s')u(s') ds \rangle ds \\
& = 2 \im \int_{0}^{t} \langle e^{-isH} u_0 , V(s)u(s) \rangle ds .
\end{align*}
Then we have $\|U(t, 0)u_0\|_2 = \|u_0\|_2$.
\end{proof}
Now we prove Theorem \ref{11102301}.
\begin{proof}[Proof of Theorem \ref{11102301}]

The proof is divided into two parts.

\noindent
(Step 1) In this step we prove the local existence of a solution. The proof follows the argument in \cite{FS}. Let $R >0$ be such that $\|\gamma_0\|_{\mathfrak{S}^{\frac{2q}{q+1}}} < R$ and $T =T(R)$ be chosen later. Set
\[ X := \left\{ (\gamma , \rho) \in C([0, T]; \mathfrak{S}^{\frac{2q}{q+1}}) \times L^p ([0, T]; L^q _x) \mid \|\gamma\|_{C([0, T]; \mathfrak{S}^{\frac{2q}{q+1}})} + \|\rho\|_{L^p ([0, T]; L^q _x)} \le CR \right\}.\]
Here $C>0$ is later chosen independent of $R$.
We define 
\begin{align*}
&\Phi (\gamma , \rho) =(\Phi_1 (\gamma , \rho), \rho [\Phi_1 (\gamma, \rho)]) \\
& \Phi_1 (\gamma , \rho)(t) = e^{-itH} \gamma_0 e^{itH} -i \int_{0}^{t} e^{-i(t-s)H} [w*\rho , \gamma] e^{i(t-s)H} ds .
\end{align*}
For $(\gamma , \rho) \in X$, by the orthonormal Strichartz estimates and the unitary invariance of the Schatten norms, we have
\begin{align*}
\| \Phi_1 (\gamma , \rho)\|_{C([0, T]; \mathfrak{S}^{\frac{2q}{q+1}})} &\lesssim \|\gamma_0\|_{\mathfrak{S}^{\frac{2q}{q+1}}} + 2\int_{0}^{t} \|w* \rho\|_{\infty} \|\gamma\|_{\mathfrak{S}^{\frac{2q}{q+1}}} ds \\
& \lesssim R + 2T^{1/p'} \|w\|_{q'} \|\rho\|_{L^p ([0, T]; L^q _x)} \|\gamma\|_{C([0, T]; \mathfrak{S}^{\frac{2q}{q+1}}) } \\
& \lesssim R + 2\|w\|_{q'} C^2 T^{1/p'} R^2 .
\end{align*}
By Lemma \ref{11131023}, 
\begin{align*}
\|\rho [\Phi_1 (\gamma , \rho)]\|_{L^p ([0, T]; L^q _x)} &\lesssim R + 2 \int_{0}^{T} \|w* \rho\|_{\infty} \|\gamma\|_{\mathfrak{S}^{\frac{2q}{q+1}}} ds \\
& \le R + 2\|w\|_{q'} C^2 T^{1/p'} R^2 .
\end{align*}
If we take $C$ large enough and $T$ small enough, we have $\Phi : X \rightarrow X$. By the similar computation, $\Phi$ is a contraction mapping on $X$. Hence there exists a unique solution to (H) on $[0, T_{max} )$.

\noindent
(Step 2)  In this step we prove the global existence of the solution. Note that if $T_{max} < \infty$ holds, then $\varlimsup_{t \nearrow T_{max} } \|\gamma (t)\|_{\mathfrak{S}^{\frac{2q}{q+1}}} = \infty$ holds. By Lemma \ref{1114032}, Corollary \ref{1114035} and $\rho_{\gamma} *w \in L^p _{loc} ( [0, T_{max} ); L^{\infty} _x)$, for every $\epsilon>0$, there exists a family of unitary operators $\{U(t)\}$ such that 
\begin{itemize}

\item $U(t)$ is strongly continuous

\item $U(0) = I$

\item $U(t)u_0 = e^{-itH} u_0 -i \int_{0}^{t} e^{-i(t-s)H} w* \rho_{\gamma}(s)U(s)u_0 ds$ on $[0, T_{max} - \epsilon ]$. 
\end{itemize}

Now we claim $\gamma (t) = U(t) \gamma_0 U(t)^* (=: \eta (t))$ holds for $t \in [0, T_{max} - \epsilon ]$. If this equality holds, we have $\|\gamma (t)\|_{\mathfrak{S}^{\frac{2q}{q+1}}} = \|\gamma_0\|_{\mathfrak{S}^{\frac{2q}{q+1}}}$ for $t \in [0, T_{max} - \epsilon ]$ by the unitary invariance of the Schatten norms. Hence $T_{max} = \infty$.
Since
\begin{align*}
&U(t) = e^{-itH} -i \int_{0}^{t} e^{-i(t-s)H} w* \rho_{\gamma}(s)U(s) ds \\
&U(t)^* = e^{itH} + i \int_{0}^{t} U(s)^* w* \rho_{\gamma}(s) e^{i(t-s)H} ds
\end{align*}
hold, we have
\begin{align*}
U(t)\gamma_0 U(t)^* = e^{-itH} \gamma_0 e^{itH} &+ \left( \int_{0}^{t} e^{-i(t-s)H} w* \rho_{\gamma}(s)U(s) ds\right) \gamma_0 \\
&\cdot\left(\int_{0}^{t} U(s')^* w* \rho_{\gamma}(s') e^{i(t-s')H} ds' \right) \\
& + 2\im \left( \int_{0}^{t} e^{-i(t-s)H} w* \rho_{\gamma}(s)U(s) \gamma_0 e^{itH} ds \right).
\end{align*}
For the second term, we transform it as
\begin{align*}
&\left( \int_{0}^{t} e^{-i(t-s)H} w* \rho_{\gamma}(s)U(s) ds\right) \gamma_0 \left(\int_{0}^{t} U(s')^* w* \rho_{\gamma}(s') e^{i(t-s')H} ds' \right) \\
& = \re \left( \int_{0}^{t} \int_{0}^{t} e^{-i(t-s)H} w* \rho_{\gamma}(s)U(s) \gamma_0 U(s')^* w* \rho_{\gamma}(s') e^{i(t-s')H} ds' ds \right) \\
& = 2 \re \left( \int_{0}^{t} \int_{0}^{s} e^{-i(t-s)H} w* \rho_{\gamma}(s)U(s) \gamma_0 U(s')^* w* \rho_{\gamma}(s') e^{i(t-s')H} ds' ds \right) \\
& = 2 \re \left( \int_{0}^{t} e^{-i(t-s)H} w* \rho_{\gamma}(s)U(s) \gamma_0 \left( \int_{0}^{s} U(s')^* w* \rho_{\gamma}(s') e^{i(s-s')H} ds' \right) e^{i(t-s)H} ds\right) \\
& = 2 \re \left( \int_{0}^{t} e^{-i(t-s)H} w* \rho_{\gamma}(s)U(s) \gamma_0 \left( \frac{U(s)^* - e^{isH}}{i} \right) e^{i(t-s)H} ds\right) \\
&= 2 \re \left( \frac{1}{i} \int_{0}^{t} e^{-i(t-s)H} w* \rho_{\gamma}(s) \eta (s) e^{i(t-s)H} ds \right) - 2\im \left( \int_{0}^{t} e^{-i(t-s)H} w* \rho_{\gamma}(s)U(s) \gamma_0 e^{itH} ds \right) \\
& = -i \int_{0}^{t} e^{-i(t-s)H} [w* \rho_{\gamma} (s) , \eta (s)] e^{i(t-s)H} ds - 2\im \left( \int_{0}^{t} e^{-i(t-s)H} w* \rho_{\gamma}(s)U(s) \gamma_0 e^{itH} ds \right) .
\end{align*}
Then we have
\[ \eta(t) = U(t)\gamma_0 U(t)^* = e^{-itH} \gamma_0 e^{itH} -i \int_{0}^{t} e^{-i(t-s)H} [w* \rho_{\gamma} (s) , \eta (s)] e^{i(t-s)H} ds .\]
By taking a difference of this equation and 
\[ \gamma (t) = e^{-itH} \gamma_0 e^{itH} -i \int_{0}^{t} e^{-i(t-s)H} [w*\rho_{\gamma} , \gamma] e^{i(t-s)H} ds,\]
we have
\begin{align*}
\|\gamma (t) - \eta(t)\|_{\mathfrak{S}^{\frac{2q}{q+1}}} &\le \int_{0}^{t} \| [w*\rho_{\gamma} , \gamma - \eta]\|_{\mathfrak{S}^{\frac{2q}{q+1}}} ds \\
& \le \int_{0}^{t} 2\|w\|_{q'} \|\rho_{\gamma} (s)\|_{q} \|\gamma (s) - \eta(s)\|_{\mathfrak{S}^{\frac{2q}{q+1}}} ds .
\end{align*}
By Gronwall's inequality, we have 
\[\gamma(t) = \eta(t) \] 
for $t \in [0, T_{max} - \epsilon]$. Hence we are done.
\end{proof}

\section*{Acknowledgement}
The author thanks his supervisor Kenichi Ito for valuable comments and discussions. He also thanks Sonae Hadama for suggesting this problem. He is partially supported by FoPM, WINGS Program, the University of Tokyo.

\end{document}